\documentclass[conference,10pt]{IEEEtran}
\usepackage{amsmath,graphicx}
\usepackage{pgf,tikz,pgfplots}

\usepackage{amsmath,amssymb,amsthm}
\usepackage{algorithmic}
\usepackage[linesnumbered,boxed]{algorithm2e}
\usepackage{bm}
\usepackage{color}
\newtheorem{theorem}{Theorem}
\usepackage{caption}

\newtheorem{lemma}{Lemma}

\newtheorem{proposition}{Proposition}

\newcommand{\C}{\mathbb{C}}
\newcommand{\bP}{\mathbf{P}}
\newcommand{\bD}{\mathbf{D}}
\newcommand{\bR}{\mathbf{R}}

\newcommand{\bg}{\mathbf{g}}
\newcommand{\bU}{\mathbf{U}}

\newcommand{\bv}{\mathbf{v}}
\newcommand{\bV}{\mathbf{V}}

\newcommand{\bI}{\mathbf{I}}

\newcommand{\bn}{\mathbf{n}}

\newcommand{\bx}{\mathbf{x}}

\newcommand{\bs}{\mathbf{s}}
\newcommand{\by}{\mathbf{y}}
\newcommand{\bz}{\mathbf{z}}

\newcommand{\bB}{\mathbf{B}}

\newcommand{\bH}{\mathbf{H}}

\newcommand{\bQ}{\mathbf{Q}}

\usepackage{subfigure}

\def\notvartualgraph{1}
\ifthenelse{\notvartualgraph=1}{

}
{

}

\begin{document}
\title{Linear One-Bit Precoding in Massive MIMO: Asymptotic SEP Analysis and Optimization}
\author{\IEEEauthorblockN{Zheyu Wu\IEEEauthorrefmark{1}\IEEEauthorrefmark{2},
		Junjie Ma\IEEEauthorrefmark{1},
		Ya-Feng Liu\IEEEauthorrefmark{1}, and
		A. Lee Swindlehurst\IEEEauthorrefmark{3}
	}
	\IEEEauthorblockA{\IEEEauthorrefmark{1}LSEC, ICMSEC, AMSS, Chinese Academy of Sciences, Beijing, China}
	\IEEEauthorblockA{\IEEEauthorrefmark{2}School of Mathematical Sciences, University of Chinese Academy of Sciences, Beijing, China}
	\IEEEauthorblockA{\IEEEauthorrefmark{3}Center for Pervasive Communications and Computing, University of California, California, United States}
	\small{Email: \{wuzy,~majunjie,~yafliu\}@lsec.cc.ac.cn, swindle@uci.edu}}

\setlength{\abovedisplayskip}{0.08cm}
\setlength{\belowdisplayskip}{0.08cm}
\setlength{\jot}{0.08cm}

\maketitle

\begin{abstract}
This paper focuses on the analysis and optimization of a class of linear one-bit precoding schemes for a downlink massive MIMO system under Rayleigh fading channels. The considered class of linear one-bit precoding is fairly general, including the well-known matched filter (MF) and zero-forcing (ZF) precoding schemes as special cases.
Our analysis is based on an asymptotic framework where the numbers of transmit antennas and users in the system grow to infinity with a fixed ratio. We show that, under the asymptotic assumption, the symbol error probability (SEP) of the considered linear one-bit precoding schemes converges to that of a scalar ``signal plus independent Gaussian noise" model. This result enables us to provide accurate predictions for the SEP of linear one-bit precoding. Additionally, we also derive the optimal linear one-bit precoding scheme within the considered class based on our analytical results.  Simulation results demonstrate the excellent accuracy of the SEP prediction and the optimality of the derived precoder.
\end{abstract}


\vspace{-0.2cm}
\section{Introduction}
\vspace{-0.1cm}
 One-bit precoding has emerged as a promising technique for reducing the complexity and hardware costs of massive multiple-input multiple-output (MIMO) systems. This innovative approach involves using one-bit digital-to-analog converters (DACs) at the base station (BS), which greatly reduces the cost and power consumption of the DACs. Additionally, the constant envelope  nature of one-bit signals also enables the use of the most power{\color{black}-}efficient power amplifiers (PAs), resulting in significant energy savings for massive MIMO systems.

One-bit precoding schemes can be broadly classified into two categories: linear and nonlinear one-bit precoding.  In the linear one-bit precoding scheme, the transmit signal is obtained by simply quantizing the output of a linear precoder. This type of precoding is favorable for its simplicity and low computational complexity 
\cite{SQUID}--\cite{ZF}. 
In contrast,  nonlinear one-bit precoding employs a nonlinear mapping from the data symbol vector that we wish to recover at the receiver to the one-bit signals transmitted at the BS, and is typically obtained by solving a difficult optimization problem. Nonlinear one-bit precoding can generally achieve better symbol error rate (SER) performance than linear one-bit precoding. However, its computational complexity is also much higher, making it challenging to implement in practical systems 
\cite{SQUID},  \cite{CImodel}--\cite{conference}. 
In this paper, we focus on the analysis and optimization of the more practical linear one-bit precoder.

The simplest linear one-bit  precoder simply quantizes the output of classical linear approaches such as matched filter (MF) and zero-forcing (ZF) precoding \cite{SQUID}. In 
\cite{WFQ}--\cite{duality},
the authors  took the one-bit quantization into account and designed efficient algorithms to obtain linear one-bit precoders based on different criteria, e.g., mean square error (MSE) or signal-to-quantization-plus-interference-plus-noise ratio (SQINR). 
In addition to precoding design, there has also been growing research interest in the performance analysis of linear one-bit  precoding \cite{SQUID,MFrate,ZF}.  In particular, the authors in \cite{ZF} gave an asymptotic analysis for one-bit ZF precoding, deriving a closed-form expression for SEP under the assumption that the numbers of transmit antennas and users tend to infinity with a fixed ratio.
The analysis in \cite{ZF} is based on the Bussgang decomposition technique \cite{Bussgang}, which transforms the non-linear quantization model into a linear one consisting of a signal term and an uncorrelated distortion term.  
However, the distribution of the uncorrelated distortion term is unknown, and the analysis in \cite{ZF} assumes, heuristically, that the distortion term {\color{black} is independent of all other random variables in the system}. A rigorous justification for such a Bussgang-decomposition-based result is still lacking.

In this paper, we focus on the analysis and optimization of a wide class of linear one-bit precoders for a downlink massive MIMO system. We derive \emph{closed-form} SEP formulas for general linear one-bit precoding, under the assumption that the numbers of transmit antennas and users in the system asymptotically grow to infinity with a fixed ratio. Based on the analytical results, we also derive the \emph{optimal} linear one-bit precoder within the considered class.  
Unlike existing analyses that are based on the Bussgang decomposition, we develop a novel analytical framework that exploits a recursive characterization of Haar random matrices. Our analytical framework is general and could potentially be useful for other applications.

\vspace{-0.15cm}
\section{{\color{black}System model and} Problem Formulation}\label{formulation}
\vspace{-0.15cm}
Consider a one-bit massive MIMO system in which an $N$-antenna BS equipped with one-bit DACs simultaneously serves $K$ single-antenna users, where $K<N$. 
Assuming perfect channel state information (CSI) at the BS and infinite resolution analog-to-digital converters (ADCs) at the user side as in \cite{SQUID}--\cite{conference}, the received signal vector at the users, denoted by $\by\in\C^K$, can be modeled as 
$$\by=\bH\bx+\bn,$$
where $\bx\in\C^N$ is the transmit signal vector from the BS, ~$\bH\in\C^{K\times N}$ is the channel matrix between the BS and the users, and $\bn$ is the additive noise. As one-bit DACs are employed, the transmit signal vector must satisfy $\bx\in\{\pm\frac{\sqrt{2}}{2}\pm \frac{\sqrt{2}}{2}j\}^N,$ where unit transmit power is assumed at each antenna. We further assume that each element of $\bH$ is independently and identically distributed (i.i.d.) following $\mathcal{CN}\left(0,\frac{1}{N}\right)$ and each element of $\bn$ is i.i.d. following $\mathcal{CN}(0,\sigma^2)$. The normalization of $\bH$ is introduced as in \cite{assumption1} to ensure that the received signal power at the users does not scale with $N$.

In a linear one-bit precoder, the transmit vector is given by 
$$\bx=q(\bP\bs).$$
Here, $\bs\in\C^{K}$ is the desired symbol vector with i.i.d. elements drawn from a QPSK constellation\footnote{The results in this paper also hold for general constellation schemes. We assume QPSK to simplify the presentation.}, $\bP\in\C^{N\times K}$ is a precoding matrix, and $q(\cdot)$ models the one-bit quantizer which acts independently on the real and imaginary components of its input vector.  
In this paper, we consider a specific class of precoding matrices. Specifically, let $\bH=\bU\bD\bV^\mathsf{H}$ be the {\color{black}singular value decomposition} (SVD) of $\bH$, where $\bU\in\mathcal{U}(K)$, $\bV\in\mathcal{U}(N)$, and $\bD=\left(\begin{matrix}\text{diag}(d_1,\dots,d_K)&\mathbf{0}_{K\times (N-K)}\end{matrix}\right)\in\mathbb{R}^{K\times N}$; $\mathcal{U}(K)$ and $\mathcal{U}(N)$ denote the sets of unitary matrices in $\C^{K}$ and $\C^N$, respectively. We focus on precoding matrices of the following form: 
\begin{equation}\label{eqn:P}
\bP=\bV f(\bD)^\mathsf{T}\bU^\mathsf{H},
\end{equation}
 where 
 $f(\cdot)$ is a positive continuous function on {\color{black}$\mathbb{R}^{++}$} that acts independently on the nonzero singular values of $\bH$, i.e., $f(\bD)=\left(\begin{matrix}\text{diag}(f(d_1),\dots,f(d_K))&\mathbf{0}_{K\times (N-K)}\end{matrix}\right)$. The reasons for considering the special structure for $\bP$ in \eqref{eqn:P} are twofold. First, the structure  in \eqref{eqn:P} is quite general and encompasses several classical precoding matrices including the MF and ZF precoders (by specifying $f(d)=d$ and $f(d)=d^{-1}$, respectively) as special cases. Second, this special structure exhibits favorable statistical properties, as will be shown in Section \ref{analysis}. 
 
With the above linear one-bit precoding scheme,  the system model becomes
\begin{equation}\label{sysmodel}
\by=\bH q(\bP\bs)+\bn=\bU\bD\bV^\mathsf{H} q(\bV f(\bD)^\mathsf{T}\bU^\mathsf{H}\bs)+\bn.
\end{equation}
At the user side, we assume as in \cite{SQUID}--\cite{conference} that symbol-wise nearest-neighbor decoding is employed,  i.e., each user $k$ maps
its received signal $y_k$ to the nearest constellation point $\hat{s}_k$. In the following, we first analyze the SEP  performance of the $K$ users, defined as
$\text{SEP}_k=\mathbb{P}\left(\hat{s}_k\neq s_k\right),$ for a given precoder (i.e., given $f$) in Section \ref{analysis} and then 
 optimize $f(\cdot)$ in terms of the SEP performance in Section \ref{optimization}.
\section{Asymptotic SEP Analysis}\label{analysis}
In this section, we derive the SEP of the considered linear one-bit precoder. Our analysis is based on an asymptotic framework commonly used in the performance analysis of massive MIMO systems (see, e.g., \cite{assumption1}). Specifically, we assume that both the number of transmit antennas and the number of users tend to infinity with a fixed ratio, i.e., $N,K\to\infty,$ $\frac{N}{K}=\gamma>1$.  We show that, under such an asymptotic assumption, the SEP of the considered linear one-bit precoding scheme
converges to 
that of a simple ``signal plus independent Gaussian noise'' model. In the following, we will first give our main result and related discussions in Section \ref{3a} and then provide a proof sketch in Section \ref{3b}.
\vspace{-0.1cm}
\subsection{Main Result}\label{3a}
\vspace{-0.1cm}
Our main result on the SEP performance of model \eqref{sysmodel} is summarized as follows.
\vspace{-0.15cm}
\begin{theorem}\label{The1}
Consider the system in Section \ref{formulation}.  As $N,K\to\infty$ with $\frac{N}{K}=\gamma>1$, we have\vspace{0.05cm}
\[
\lim\limits_{N,K\to\infty}\text{\normalfont{SEP}}_k=\overline{\text{\normalfont{SEP}}},\quad \forall\, k\in[K],
\]
where  $\overline{\text{\normalfont{SEP}}}$ is the SEP of the following asymptotic scalar model: 
\begin{equation}\label{eqn:asym}
\bar{y}=\overline{T}_s\, s+\overline{T}_{g}\,g+n.
\end{equation}
In the above model \eqref{eqn:asym}, $\{s,g,n\}$ are independent with $s$ uniformly drawn from a QPSK constellation, $g\sim\mathcal{CN}(0,1), ~n\sim\mathcal{CN}(0,\sigma^2)${\color{black};  $\overline{T}_s$ and $\overline{T}_g$ are two constants given by }
\begin{equation}\label{TsTg}
\begin{aligned}
\overline{T}_s&=\sqrt{\frac{2\gamma}{\pi\,\mathbb{E}\left[f^2(d)\right]}}\,\mathbb{E}[d\,f(d)],\\
\overline{T}_g&=\sqrt{\frac{2\gamma\,\text{\normalfont{var}}[d\,f(d)]}{\pi\,\mathbb{E}[f^2(d)]}+1-\frac{2}{\pi}},
\end{aligned}
\end{equation}
{\color{black}where} $d=\sqrt{\lambda}$ and $\lambda$ follows the Marchenko-Pastur distribution, whose probability density function is 
\begin{equation*}\label{MPdistribution}
p_{\lambda}(x)=\frac{\sqrt{(x-a)_+(b-x)_+}}{2\pi cx}
\end{equation*} with 
$a=(1-\sqrt{c})^2, b=(1+\sqrt{c})^2, c=\frac{1}{\gamma}$, and $(x)_+=\max\{x,0\}.$ 
\end{theorem}\vspace{-0.1cm}
Theorem \ref{The1} shows that the SEP of model \eqref{sysmodel} can be asymptotically characterized by the simple scalar model given in \eqref{eqn:asym}. It is well-known that the SEP of  \eqref{eqn:asym} 
can be tightly approximated as 
\begin{equation}\label{sep}\overline{\text{SEP}}\approx2Q\left(\sqrt{\overline{\text{SNR}}}\right),
\end{equation}
where $Q(x)=\frac{1}{\sqrt{2\pi}}\int_x^\infty e^{-\frac{1}{2}t^2}dt$ and $\overline{\text{SNR}}$ is the signal-to-noise-ratio (SNR) of \eqref{eqn:asym}, which is  given by \vspace{0.05cm}
\begin{equation}\label{snr}
\overline{\text{SNR}}=\frac{\overline{T}_s^2}{\overline{T}_g^2+\sigma^2}=\frac{\mathbb{E}^2[d\,f(d)]}{\text{var}[d\,f(d)]+\frac{1-\frac{2}{\pi}+\sigma^2}{\frac{2}{\pi}\gamma}\,\mathbb{E}[f^2(d)]}.\vspace{0.05cm}
\end{equation}
The above equations \eqref{sep} and \eqref{snr} clearly demonstrate the effect of the function $f(\cdot)$ on the SEP performance. In particular, by specifying $f(\cdot)$ as $f(x)=x$ or $f(x)=x^{-1}$, we can respectively obtain the following asymptotic SNR formulas for the one-bit MF and ZF precoders:\vspace{0.1cm}
\begin{equation*}
\overline{\text{SNR}}_{\text{MF}}= \frac{\frac{2}{\pi}\gamma}{1+\sigma^2},~~\overline{\text{SNR}}_{\text{ZF}}= \frac{\frac{2}{\pi}(\gamma-1)}{1-\frac{2}{\pi}+\sigma^2},
\end{equation*}
and the corresponding SEP formulas can be obtained by substituting the above expressions into \eqref{sep}. 
We remark here that an identical SEP formula for one-bit ZF precoding was derived in \cite{ZF} using the Bussgang decomposition technique ({\color{black}which involves a heuristic step to treat the  distortion term as a random variable  independent of the others; see the discussions in Section I}). Our Theorem \ref{The1} provides the SEP formula for a more general linear one-bit precoder using a {\color{black} rigorous} analytical technique as detailed in the next subsection. 
\vspace{-0.1cm}
\subsection{Proof Outline}\label{3b}
\vspace{-0.05cm}
In this subsection, we outline the proof of Theorem \ref{The1}. Our proof consists of two main steps. First, we derive a statistically equivalent model of \eqref{sysmodel} that is approximately in a signal-plus-independent-Gaussian-noise form and is more amenable to analysis; see Proposition \ref{Pro1} below. This step is non-asymptotic and the main technique is the use of Householder dice (HD) \cite{HD}, a novel approach for recursively generating Haar random matrices; see Lemma \ref{Haar} and the discussions that follows.
Second, we give the asymptotic analysis, which shows that, under the asymptotic assumption in Theorem 1, the statistically equivalent model converges to a simple signal-plus-independent-Gaussian-noise model. This statement will be made precise in Proposition \ref{Pro2}. Next we give more details on these two main steps.
\subsubsection{Statistically Equivalent Model}
In this part, we derive a statistically equivalent model for \eqref{sysmodel}. We begin with a well-known result in random matrix theory \cite{tulino2004random}, which characterizes the distributions of the random matrices appearing in \eqref{sysmodel}.\vspace{-0.15cm}
\begin{lemma}
Let $\bH=\bU\bD\bV^\mathsf{H}$ be the SVD of $\bH$, where the entries of $\bH$ are i.i.d. following $\mathcal{CN}(0,\frac{1}{N})$. Then $\bU,\bD, \bV$ are independent, and $\bU$ and $\bV$ are Haar distributed random matrices in $\mathcal{U}(K)$ and $\mathcal{U}(N)$, respectively, i.e., \vspace{0.05cm}
$$\bQ_1\bU\overset{d}{=}\bU\overset{d}{=}\bU\bQ_1,~\forall~ \bQ_1\in\mathcal{U}(K),$$
$$\bQ_2\bV\overset{d}{=}\bV\overset{d}{=}\bV\bQ_2,~\forall ~\bQ_2\in\mathcal{U}(N),\vspace{-0.15cm}$$
{\color{black}where $\overset{d}{=}$ means equal in distribution}.
\end{lemma}\vspace{-0.2cm}
The above lemma suggests that all random matrices/vectors in model \eqref{sysmodel} are mutually independent, and $\bU$ and $\bV$ are Haar distributed.  
 To deal with these Haar random matrices, we apply the HD technique proposed in \cite{HD}.
 Before going into details, we first define a unitary reflector as in \cite[{\color{black}Eq. 27}]{HD}:  for any nonzero vector $\bv\in\mathbb{C}^{\color{black}m}$, define
$$\bR(\bv) = (-e^{-j\theta})\left(\bI_{\color{black}m}-\frac{\left(\frac{\bv}{\|\bv\|}+e^{j\theta}\mathbf{e}_1\right)\left(\frac{\bv}{\|\bv\|}+e^{j\theta}\mathbf{e}_1\right)^\mathsf{H}}{1+r}\right),$$
where $\frac{v_1}{\|\bv\|}=re^{j\theta}$ with $r\geq 0$ (when $v_1 = 0$, we set $\theta=0$) and $\mathbf{e}_1=[1,0,\dots,0]^\mathsf{T}$.  
It is easy to check that $\bR(\bv)$ is unitary and satisfies
\begin{equation}\label{propertyR}
\bR(\bv)\mathbf{e}_1=\frac{\bv}{\|\bv\|},~~\bR(\bv)^\mathsf{H}\bv=\|\bv\|\mathbf{e}_1.
\end{equation}
 In fact, $\bR(\bv)$ is a rotation {\color{black}of} the Householder transform of $\bv$. We also define the generalized reflector 
 $$\bR_k(\bv)=\left(\begin{matrix}\bI_{k-1}&\mathbf{0}\\\mathbf{0}&\bR(\bv[k:{\color{black}m}])\end{matrix}\right),~{\color{black}1\leq k\leq {\color{black}m}},$$
{\color{black}where $\bv[k:m]$ denotes the $k$-th to {\color{black}$m$}-th elements of $\bv$.}
 
Lemma \ref{Haar} below gives a recursive characterization of the Haar matrices and is the theoretical basis of the HD technique.  It is a direct generalization of \cite[Lemma 1]{HD} from the real space to the complex space. 

 \begin{lemma}\label{Haar}
Let $\mathbf{g}\sim\mathcal{C}\mathcal{N}(\mathbf{0},\mathbf{I}_{\color{black}m})$, $\bQ_{{\color{black}m}-1}\sim\text{\normalfont{Haar}}({\color{black}m}-1)$, and $\bv\in\C^n\backslash\{\mathbf{0}\}$, all of which are independent. Then
\begin{equation*}\label{construct}
\bQ_{{\color{black}m}}:=\bR_1(\mathbf{g})\left(\begin{matrix}1&\mathbf{0}\\\mathbf{0}&\bQ_{{\color{black}m}-1}\end{matrix}\right)\bR_1(\bv)^\mathsf{H}\sim\text{\normalfont{Haar}}({\color{black}m}),
\end{equation*}
\begin{equation*}\label{construct2}
\widetilde{\bQ}_{{\color{black}m}}:=\bR_1(\bv)\left(\begin{matrix}1&\mathbf{0}\\\mathbf{0}&\bQ_{{\color{black}m}-1}\end{matrix}\right)\bR_1(\bg)^\mathsf{H}\sim\text{\normalfont{Haar}}({\color{black}m}),
\end{equation*}
where $\text{\normalfont{Haar}}({\color{black}m})$ denotes the Haar distribution in $\mathcal{U}({\color{black}m})$. Moreover, $\bQ_{\color{black}m}$ and $\widetilde{\bQ}_{\color{black}m}$ 
are independent of $\bv$.
 \end{lemma}
 
 Now we are ready to apply the HD technique to our model \eqref{sysmodel}. The HD technique was originally proposed as an efficient algorithm for performing an iterative process involving large Haar matrices \cite{HD}. Here we use it as a powerful tool for analysis. The main idea is to recursively generate the Haar matrices $\bU$ and $\bV$ in \eqref{sysmodel}, using Lemma \ref{Haar}, in such a way that the resulting model is  amenable to analysis. 

To begin, we rewrite \eqref{sysmodel} as the following iterative form:
 \begin{equation*}\label{Eqn:model_s1s3}
\begin{aligned}
\bs_1&=f(\bD)^\mathsf{T}\bU^\mathsf{H}\bs,~~~\bs_2=q(\bV \bs_1),\\
\bs_3&=\bD\bV^\mathsf{H}\bs_2,\qquad~~ \hspace{0.05cm}\by=\bU\bs_3+\bn.
\end{aligned}
\end{equation*}
We first construct a Haar matrix $\bU^{(1)}$ according to Lemma \ref{Haar}:
\begin{equation*}\label{U}
\bU^{(1)}=\bR_1(\bs)\left(\begin{matrix}1&\mathbf{0}\\\mathbf{0}&\bQ_{K-1}\end{matrix}\right)\bR_1(\mathbf{g}_1)^\mathsf{H}, 
\end{equation*}
where $\bg_1\sim\mathcal{C}\mathcal{N}(\mathbf{0},\mathbf{I}_K)$ and $\bQ_{K-1}\sim\text{Haar}(K-1)$ are  independent and independent of $\bs$, $\bD$, and $\bn$. Then we have 
$$
\begin{aligned}
\tilde{\bs}_1&=f(\bD)^\mathsf{T}\left(\bU^{(1)}\right)^\mathsf{H}\bs\\
&=f(\bD)^\mathsf{T}\bR_1(\bg_1)\left(\begin{matrix}1&\mathbf{0}\\\mathbf{0}&\bQ_{K-1}^\mathsf{H}\end{matrix}\right)\bR_1(\bs)^\mathsf{H}\bs\\
&=f(\bD)^\mathsf{T}\bR_1(\bg_1)\bR_1(\bs)^\mathsf{H}\bs,
\end{aligned}$$
where the last equality holds since $\bR_1(\bs)^\mathsf{H}\bs$ {\color{black} is only non-zero in its first element due to \eqref{propertyR}}. At the second iteration, we generate the  Haar matrix $\bV^{(1)}$ according to Lemma \ref{Haar} as 
\begin{equation*}\label{eq:V1}
\bV^{(1)}=\bR_1(\bz_1)\left(\begin{matrix}1&\mathbf{0}\\\mathbf{0}&\bP_{N-1}\end{matrix}\right)\bR_1(\tilde{\bs}_1)^\mathsf{H},
\end{equation*}
where $\bz_1\sim\mathcal{C}\mathcal{N}(\mathbf{0},\mathbf{I}_N)$ and $\bP_{N-1}\sim\text{Haar}(N-1)$ are  independent and independent of all the existing random variables. 
It follows that $
\tilde{\bs}_2=q(\bV^{(1)}\tilde{\bs}_1){=}q(\bR_1(\bz_1)\bR_1(\tilde{\bs}_1)^\mathsf{H}\tilde{\bs}_1).
$  
At the third iteration, we define  $\bv_1=\bR_1(\bz_1)^\mathsf{H}\tilde{\bs}_2$ and further construct $\bP_{N-1}$ in $\bV^{(1)}$ as 
$$\bP_{N-1}=\bR_1(\bv_1[2:N])\left(\begin{matrix}1&\mathbf{0}\\\mathbf{0}&\bP_{N-2}\end{matrix}\right)\bR_1(\bz_2[2:N])^\mathsf{H},$$ where 
 $\bz_2\sim\mathcal{C}\mathcal{N}(\mathbf{0},\mathbf{I}_N)$ and $\bP_{N-2}\sim\text{Haar}(N-2)$ are independent and independent of all the existing random variables. Then we have 
\begin{equation*}\label{eqn:V2}
\bV^{(2)}=\bR_1(\bz_1)\bR_2(\bv_1)\left(\begin{matrix}\mathbf{I}_2&\mathbf{0}\\\mathbf{0}&\bP_{N-2}\end{matrix}\right)\bR_2(\bz_2)^\mathsf{H}\bR_1(\tilde{\bs}_1)^\mathsf{H}
\end{equation*}
 and 
 $$
 \begin{aligned}
 \tilde{\bs}_3&=\bD \left(\bV^{(2)}\right) ^\mathsf{H}\tilde{\bs}_2=\bD \bR_1(\tilde{\bs}_1)\bR_2(\bz_2)\bR_2(\bv_1)^\mathsf{H}\bv_1. 
 \end{aligned}
 $$
 Finally,  let $\bv_2=\bR_1(\bg_1)^\mathsf{H}\tilde{\bs}_3$ and construct $\bQ_{K-1}$ in $\bU^{(1)}$ as
 $$\bQ_{K-1}=\bR_1(\bg_2[2:K])\left(\begin{matrix}1&\mathbf{0}\\\mathbf{0}&\bQ_{K-2}\end{matrix}\right)\bR_1(\bv_2[2:K])^\mathsf{H},$$ 
where $\bg_2\sim\mathcal{C}\mathcal{N}(\mathbf{0},\mathbf{I}_K)$ and  $\bQ_{K-2}\sim\text{Haar}(K-2)$ are independent and independent of all the existing random variables.  Then we have 
\begin{equation*}\label{eqn:U2}\bU^{(2)}=\bR_1(\bs)\bR_2(\bg_2)\left(\begin{matrix}\mathbf{I}_2&\mathbf{0}\\\mathbf{0}&\bQ_{K-2}\end{matrix}\right)\bR_2(\bv_2)^\mathsf{H}\bR_1(\bg_1)^\mathsf{H}
\end{equation*} and 
\begin{equation}\label{tildey}
\begin{aligned}
\tilde{\by}&=\bU^{(2)}\tilde{\bs}_3+\bn=\bR_1(\bs)\bR_2(\bg_2)\bR_2(\bv_2)^\mathsf{H}\bv_2+\bn.
\end{aligned}
\end{equation}

So far, we have obtained a new model \eqref{tildey} using the HD technique. Similar to \cite{HD}, we can prove that $(\tilde{\by},\bs)$ in \eqref{tildey} and $(\by,\bs)$ in \eqref{sysmodel} are statistically equivalent.  {\color{black}Using the properties of $\bR(\cdot)$ in \eqref{propertyR}, we further get the following statistically equivalent model of \eqref{tildey}}, whose proof is omitted due to the limited space.
 \begin{proposition}[Statistically Equivalent Model]\label{Pro1}
The distribution of $(\by,\bs)$ in \eqref{sysmodel} is the same as $(\hat{\by},\bs)$ specified by the following model:\vspace{-0.05cm}\begin{equation}\label{Equiy}
\begin{aligned}
 \hat{\by}= {T}_s\,\bs+ T_g\, \bg_2+\bn,
 \end{aligned}
\end{equation}
where
\begin{equation*}\label{Equiy_2}
\begin{aligned}
T_s\hspace{-0.05cm}=&\frac{\bg_1^\mathsf{H}\{C_1\,\bD\tilde{\bs}_1\hspace{-0.05cm}+\hspace{-0.05cm}C_2\,\bD \bB(\tilde{\bs}_1)\bz_2[2:N]\}}{\|\bg_1\|\|\bs\|}\hspace{-0.05cm}-\hspace{-0.05cm}T_g\,\frac{(\bR(\bs)^{-\hspace{-0.05cm}1}\bg_2)[1]}{\|\bs\|},\\
T_g\hspace{-0.05cm}=&\frac{\|\bB(\bg_1)^\mathsf{H}\{C_1\bD \tilde{\bs}_1+C_2\bD \bB(\tilde{\bs}_1)\bz_2[2:N]\}\|}{\|(\bR(\bs)^{-1}\bg_2)[2:K]\|},\\
C_1=&\frac{\bz_1^\mathsf{H}q\left(\bz_1\right)}{\|\tilde{\bs}_1\|\|\bz_1\|},~ C_2\hspace{-0.02cm}=\hspace{-0.05cm}\frac{\left\|\bB(\bz_1)^\mathsf{H}q\left(\bz_1\right)\right\|}{\|\bz_2[2:N]\|}, ~\tilde{\bs}_1\hspace{-0.05cm}=\hspace{-0.05cm}\frac{\|\bs\|}{\|\bg_1\|} f(\bD)^{\mathsf{T}}\bg_1.
 \end{aligned}
 \end{equation*} 
 In the above expressions, $\bg_1, \bg_2, \bz_1,\bz_2$ are independent standard Gaussian random vectors, which are further independent of $\{\bs,\bD,\bn\}$;  $\bB(\cdot)$ represents the submatrix of $\bR(\cdot)$ with the first column removed.\vspace{-0.05cm}
\end{proposition}

\subsubsection{Asymptotic Analysis}
Although the statistically equivalent model in \eqref{Equiy} seems to admit a simple structure, $T_s$, $T_g$ and $\bs$, $\bg_2$ are correlated in a  complicated fashion. This motivates us to further consider the large system limit. In this part, we  show that under the asymptotic assumption where $N,K\to\infty$ with $\frac{N}{K}=\gamma>1$, both $T_s$ and $T_g$ converge to deterministic quantities. 

Our main asymptotic result is summarized {\color{black}as follows}. \vspace{-0.05cm}
\begin{proposition}[Asymptotic Analysis]\label{Pro2}
Let 
\begin{equation}\label{Asympmodel}
\bar{\by}:=\overline{T}_s\,\bs+\overline{T}_g\,\bg_2+\bn,
\end{equation}
where $\overline{T}_s$ and $\overline{T}_g$ are given in \eqref{TsTg}. 
Then the following holds as $N, K\to\infty$ with $\frac{N}{K}=\gamma>1$:
\[
(\hat{y}_k,s_k)\xrightarrow{a.s.}(\bar{y}_k,s_k),\quad \forall\, k\in[K],
\]
where $(\hat{y}_k,s_k)$ and $(\bar{y}_k,s_k)$  are {\color{black} the $k$-th rows of $(\hat{\by}, \bs)$ and $(\bar{\by},\bs)$} given in \eqref{Equiy} and \eqref{Asympmodel}, respectively.

\end{proposition}
\hspace{-0.34cm}\textit{Proof Sketch:~} It suffices to prove that 
$T_s\xrightarrow{a.s.}\overline{T}_s$ and $T_g\xrightarrow{a.s.}\overline{T}_g.$ 
First, according to the law of large numbers, we have 
$$\frac{\|\bg_i\|^2}{K}\xrightarrow{a.s.}1,~\frac{\|\bz_i\|^2}{N}\xrightarrow{a.s.}1,~i=1,2,~\frac{\bz_1^\mathsf{H}q(\bz_1)}{N}\xrightarrow{a.s.}\sqrt{\frac{2}{\pi}}.$$
By further noting that $\|\bs\|=\sqrt{K}$ and  using the relation $\bR(\bv)=\left(\frac{\bv}{\|\bv\|},\quad \bB(\bv)\right)$ for all $\bv\neq \mathbf{0},$ one can show that the remaining terms in $T_s$ and $T_g$ to be analyzed all have the following forms:
$$\frac{\bg_1^\mathsf{H}f_1(\bD)f_2(\bD)^\mathsf{T}\bg_1}{K}~\text{  and  }~ \frac{\bg_1^\mathsf{H}f_1(\bD)f_2(\bD)^\mathsf{T}\bg_2}{K},$$
where $f_1(\cdot)$ and $f_2(\cdot)$ are some positive continuous functions. From the law of large numbers and the Marchenko-Pastur law \cite{tulino2004random}, we can prove that 
$$\frac{\bg_1^\mathsf{H}f_1(\bD)f_2(\bD)^\mathsf{T}\bg_1}{K}\xrightarrow{a.s.} \mathbb{E}[f_1(d)f_2(d)]$$
and 
$$\frac{\bg_1^\mathsf{H}f_1(\bD)f_2(\bD)^\mathsf{T}\bg_2}{K}\xrightarrow{a.s.} 0,$$
where $d$ is defined as in Theorem \ref{The1}. This completes the proof.\qed\\

\vspace{-0.2cm}
Using Proposition \ref{Pro2} and noting that the SEP of \eqref{Asympmodel} is the same as that of the scalar asymptotic model \eqref{eqn:asym} in Theorem \ref{The1}, we can finally show that $\text{SEP}_k\rightarrow\overline{\text{SEP}}$, which gives Theorem \ref{The1}. The proof of the SEP convergence is straightforward and is omitted here.
\section{Optimal Linear One-Bit Precoding}\label{optimization}
As shown in Section \ref{analysis}, the function $f(\cdot)$ involved in defining $\bP$ has a major impact on the system SEP performance. In this section, we derive the optimal $f(\cdot)$ based on the asymptotic analysis in Theorem \ref{The1}. 
Specifically, our goal is to find $f(\cdot)$ that minimizes $\overline{\text{SEP}}$. 
It is simple to check that minimizing $\overline{\text{SEP}}$ in \eqref{sep} is equivalent to maximizing $\overline{\text{SNR}}$ in \eqref{snr}. Hence, we only need to focus on the following $\overline{\text{SNR}}$ maximization problem:
\begin{equation}\label{snrmax}
\max_{f:\mathbb{R}^{++}\rightarrow\mathbb{R}^{++}}~\frac{\mathbb{E}^2[d\,f(d)]}{\text{var}[d\,f(d)]+\frac{1-\frac{2}{\pi}+\sigma^2}{\frac{2}{\pi}\gamma}\,\mathbb{E}[f^2(d)]}.
\end{equation}
Somewhat surprisingly, the solution to the above infinite dimensional functional optimization problem has a simple structure, as shown in the following theorem.
\vspace{-0.2cm}
\begin{theorem}\label{the:optf}
The optimal solution to problem \eqref{snrmax} is a class of functions of the following form:
$$f^*(x)=\frac{\sqrt{\alpha} x}{x^2+\rho^*},\quad \text{where }\alpha>0\text{ and }\rho^*=\frac{1-\frac{2}{\pi}+\sigma^2}{\frac{2}{\pi}\gamma}.$$
\end{theorem}
\begin{proof}
Note that problem \eqref{snrmax} can be written in an equivalent form as
\begin{equation}\label{prob2}
\begin{aligned}
\min_{f:\mathbb{R}^{++}\rightarrow\mathbb{R}^{++}}~&\frac{\mathbb{E}\left[\left(d^2+\rho^*\right) f^2(d)\right]}{\mathbb{E}^2[d\, f(d)]}.
\end{aligned}
\end{equation}
Using the Cauchy-Schwarz inequality $\mathbb{E}^2[XY]\leq \mathbb{E}[X^2]\,\mathbb{E}[Y^2]$, we can upper bound $\mathbb{E}^2[d\, f(d)]$ as 
\begin{equation}\label{upperbound}
\mathbb{E}^2[d\, f(d)]\leq \mathbb{E}\left[\left(d^2+\rho^*\right) f^2(d)\right]\mathbb{E}\left[\frac{d^2}{d^2+\rho^*}\right],
\end{equation}
where  equality holds when $X=\alpha Y$, i.e., $f^2(d)=\frac{\alpha d^2}{\left(d^2+\rho^*\right)^2}$ for some $\alpha\in\mathbb{R}$. Then, plugging \eqref{upperbound} into \eqref{prob2} gives a lower bound on the objective function, i.e.,
$$\frac{\mathbb{E}\left[\left(d^2+\rho^*\right) f^2(d)\right]}{\mathbb{E}^2[d\, f(d)]}\geq \frac{1}{\mathbb{E}\left[\frac{d^2}{d^2+\rho^*}\right]},$$
which can be achieved by a class of functions satisfying $f^2(x)=\frac{\alpha x^2}{\left(x^2+\rho^*\right)^2}.$ Since $f>0$, we get $f^*$ in Theorem \ref{the:optf}.
\end{proof}

Substituting the optimal $f^*$ into  \eqref{eqn:P}, we get the optimal precoding matrix:
\begin{equation}\label{optP}
\bP_{\text{opt}}=\bH^\mathsf{H}(\bH\bH^\mathsf{H}+\rho^*\mathbf{I})^{-1},
\end{equation}
where we have ignored the positive scaling factor $\alpha$ in $f^*$ as it will be absorbed into the one-bit quantizer.
Note that $\bP_{\text{opt}}$ in \eqref{optP} is exactly the regularized ZF (RZF) precoding matrix \cite{RZF}, and the regularization parameter $\rho^*$ is determined by the system parameters $\gamma=N/K$ and $\sigma^2$. 

By specifying $f=f^*$ in \eqref{snr}, we can obtain the asymptotic SNR  for the above optimal one-bit precoder:
$$\overline{\text{SNR}}_{\text{opt}}= \frac{\sqrt{u^2+4\rho^*}+u}{2\rho^*}-1,~~u=\rho^*+\gamma-1.$$
Note that when $\gamma$ is large,  $\sqrt{u^2+4\rho^*}\approx u$ and hence 
\begin{equation}\label{ZFRZF}
\overline{\text{SNR}}_{\text{opt}}\approx \frac{u}{\rho^*}-1=\overline{\text{SNR}}_{\text{ZF}},
\end{equation}
i.e., one-bit ZF precoding is nearly optimal when the antenna-user ratio $\gamma$ is large. 
\vspace{-0.1cm}
\section{Simulation Results}
\vspace{-0.05cm}
In this section, we provide some simulation results to verify the analytical results obtained in the previous sections.

In Fig. \ref{ser1}, we validate the accuracy of the SEP prediction given in \eqref{sep}. We {\color{black}consider the noiseless case} and depict the SER performance versus $\gamma$  in Fig. {\color{black}\ref{ser1}(a)}, and fix $\gamma=6$ and depict the SER performance versus the SNR in Fig. {\color{black}\ref{ser1}(b)}. As shown in the figure, the asymptotic result in \eqref{sep} gives an accurate prediction of the SEP performance even for a small system with $K=20$. Furthermore, for a large system of $K=100$, the simulation result is almost identical to the asymptotic result. 

In Fig. \ref{ser2}, we demonstrate the optimality of the precoding scheme given in \eqref{optP} by comparing it with classical one-bit MF and ZF precoding. It can be observed that the derived RZF precoding approach exhibits the optimal SER performance for both $\gamma=2$ and $\gamma=8$.  In particular, the RZF precoder demonstrates significant superiority when $\gamma=2$, while its performance is similar to ZF when $\gamma=8$, which validates our discussion in \eqref{ZFRZF}.
\begin{figure} [t]
\centering    
  \hspace{-0.06\columnwidth}
\subfigure[SER versus $\gamma$.] { 
\includegraphics[width=0.52\columnwidth]{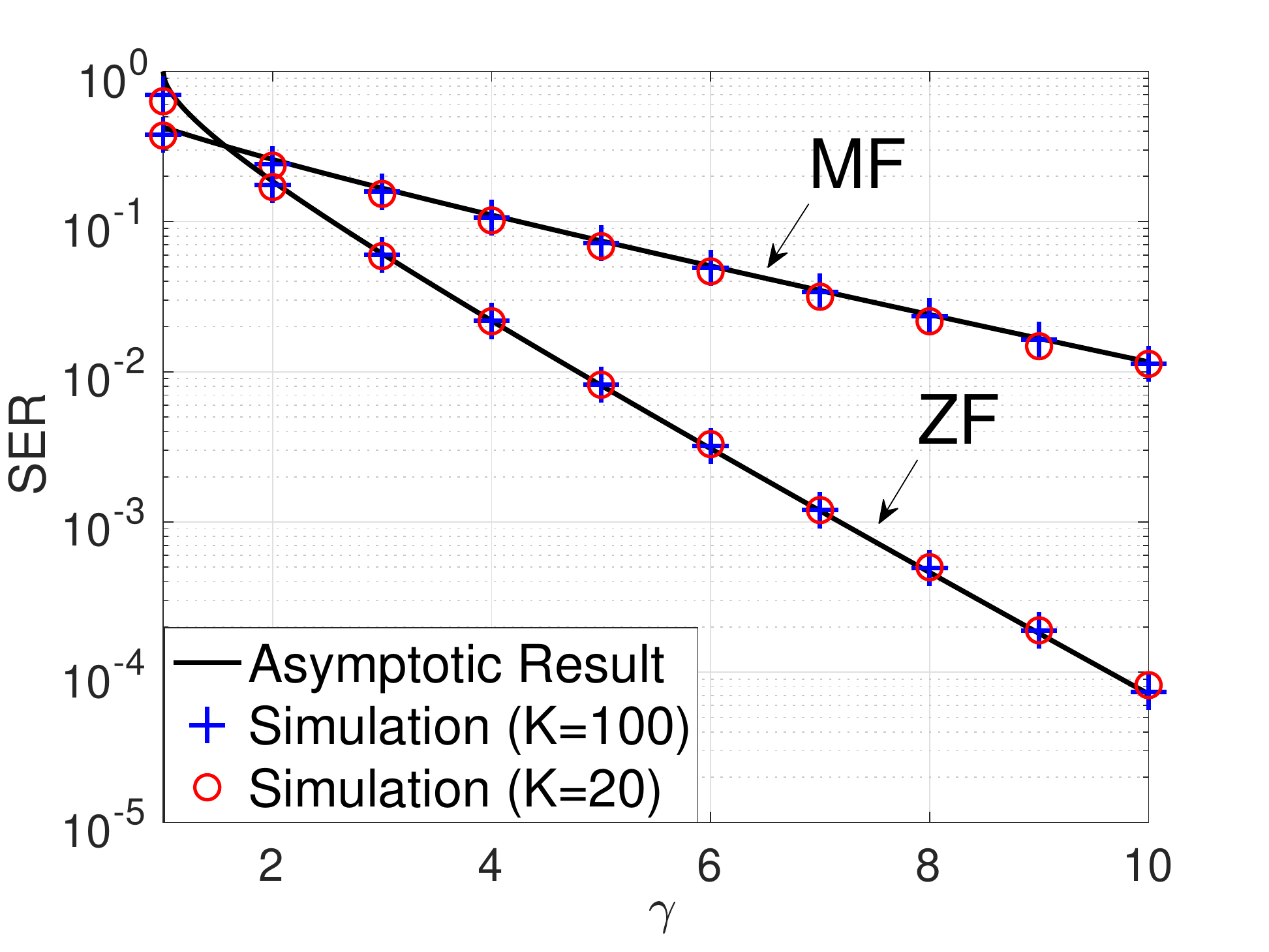}  
}     
\hspace{-0.08\columnwidth}
\subfigure[SER versus SNR.] { 
\includegraphics[width=0.52\columnwidth]{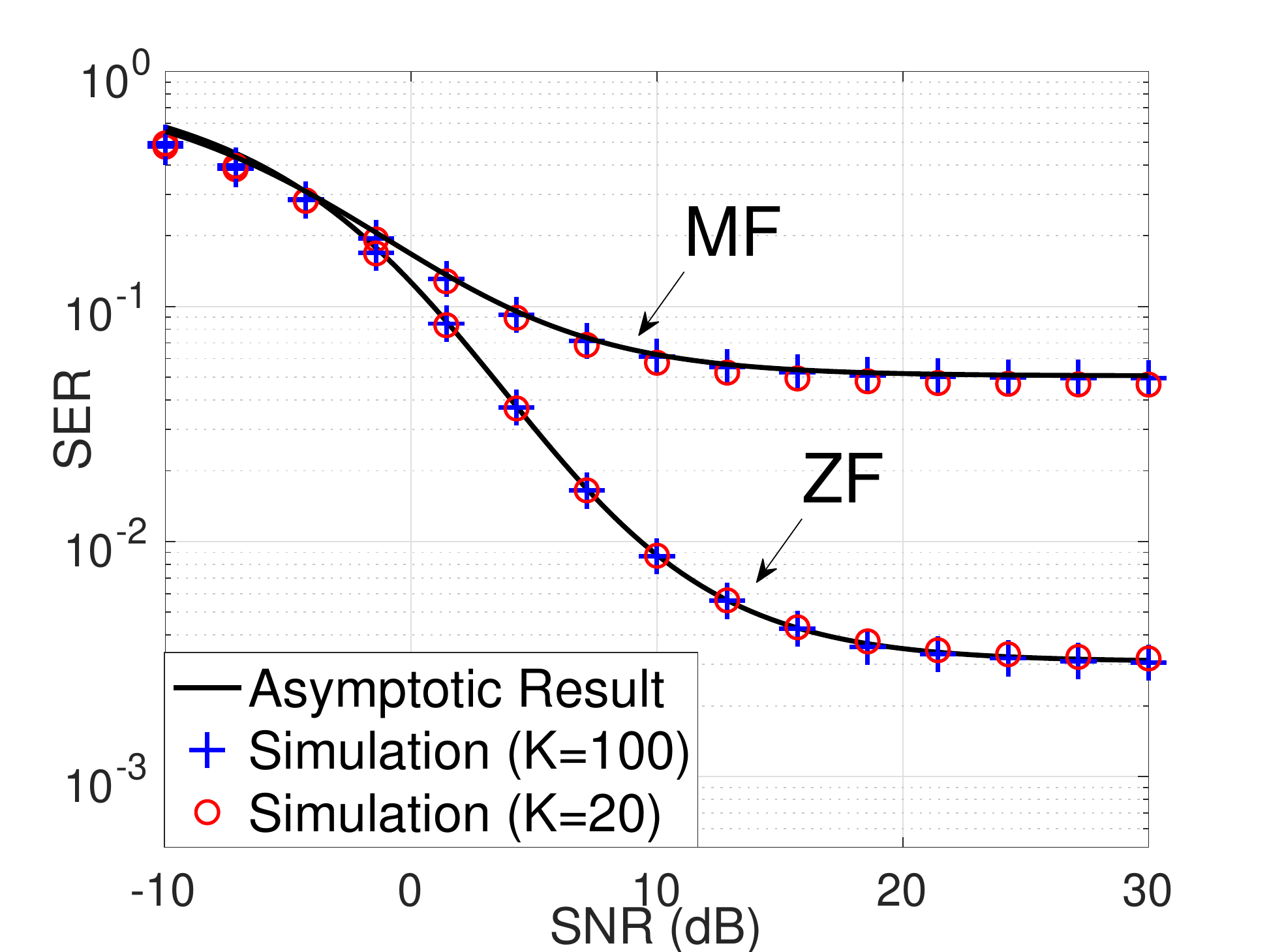}  
}    \vspace{-0.2cm}
\caption{Asymptotic and simulated SERs.}
 \label{ser1}
 \vspace{-0.1cm}
\end{figure}
 \begin{figure} [t]
\centering    
  \hspace{-0.06\columnwidth}
\subfigure[$\gamma=2$.] { 
\includegraphics[width=0.52\columnwidth]{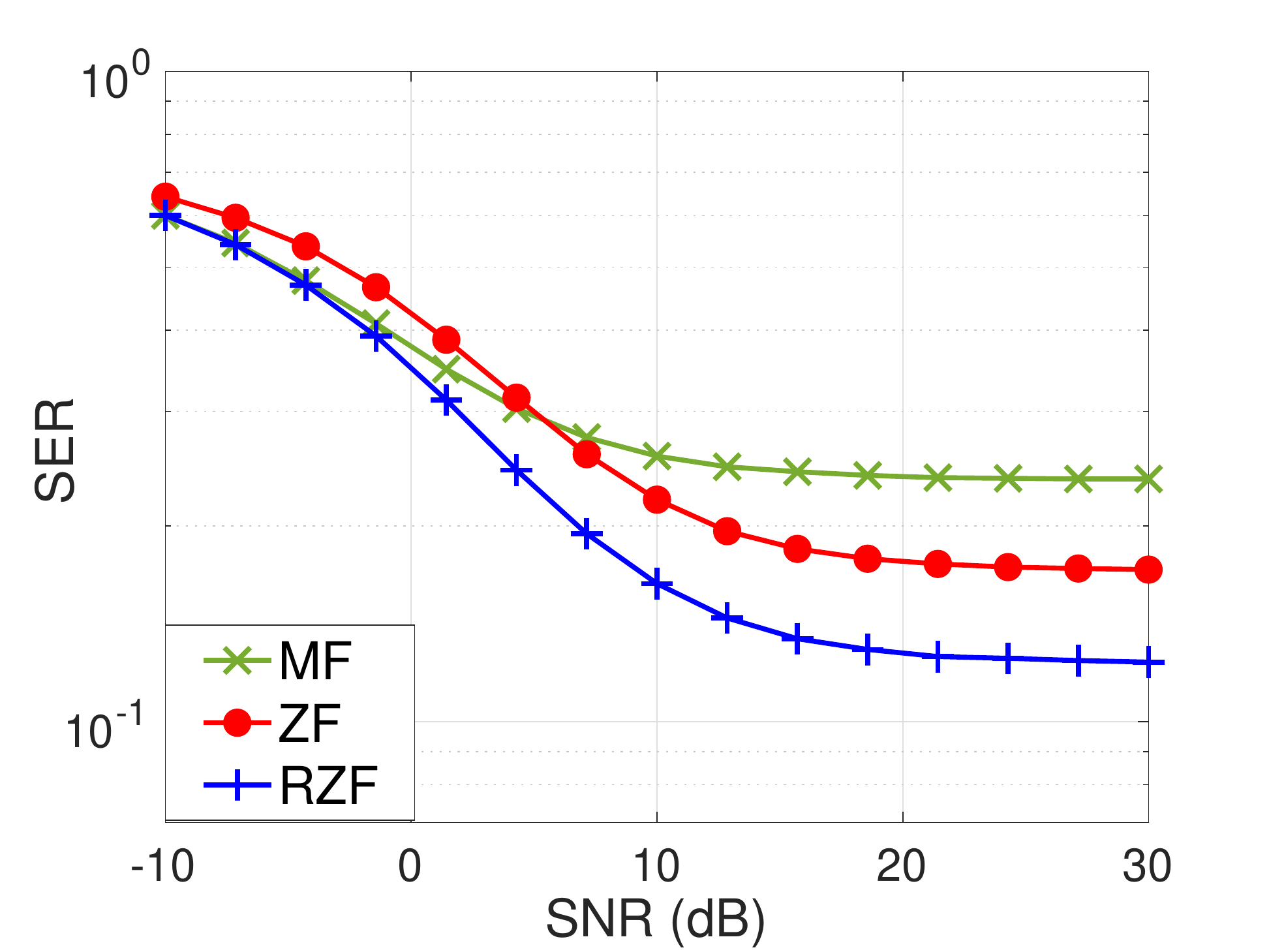}  
}     
\hspace{-0.08\columnwidth}
\subfigure[$\gamma=8$.] { 
\includegraphics[width=0.52\columnwidth]{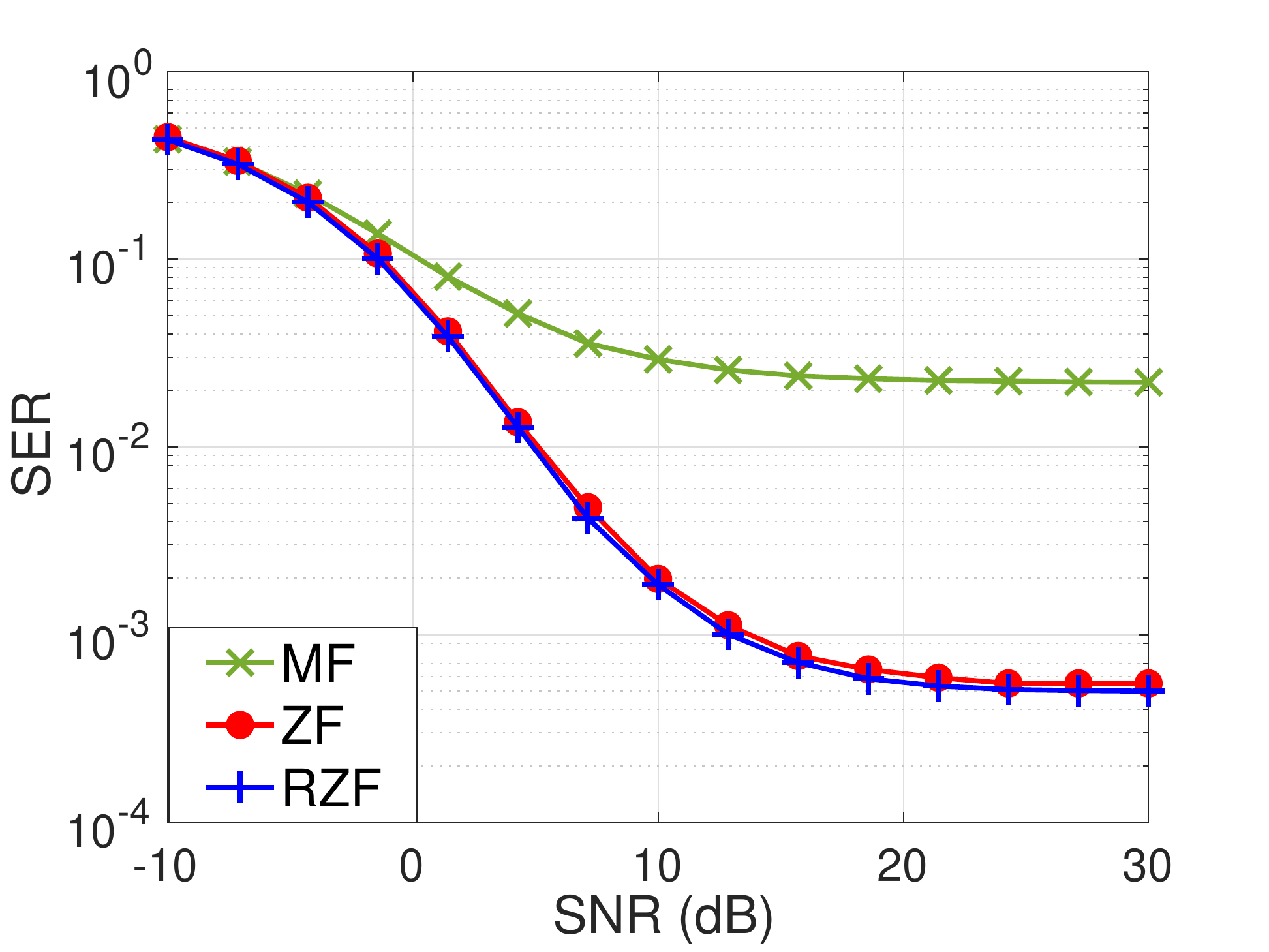}  
}    \vspace{-0.25cm}
\caption{SER performance versus SNR for different linear one-bit precoding schemes with $K=20$.}
 \label{ser2}
 \vspace{-0.4cm}
\end{figure}
\vspace{-0.35cm}
\section{{\color{black}Conclusion}}
\vspace{-0.05cm}
This paper has investigated the performance of a wide class of linear one-bit precoders for  massive MIMO systems. Through the asymptotic analysis, we have derived sharp SEP formulas for the considered linear one-bit precoders. We have also derived the optimal linear one-bit precoder, which corresponds to RZF precoding whose regularization parameter is determined by the ratio of transmit antennas to users and the variance of the additive noise. Our analysis is based on a novel and general analytical framework. An interesting future work is to  extend our framework to analyze more complicated scenarios, such as general  quantized precoding.



\vspace{-0.1cm}
\begin{thebibliography}{10}
\vspace{-0.1cm}
\providecommand{\url}[1]{#1}
\csname url@samestyle\endcsname
\providecommand{\newblock}{\relax}
\providecommand{\bibinfo}[2]{#2}
\providecommand{\BIBentrySTDinterwordspacing}{\spaceskip=0pt\relax}
\providecommand{\BIBentryALTinterwordstretchfactor}{4}
\providecommand{\BIBentryALTinterwordspacing}{\spaceskip=\fontdimen2\font plus
\BIBentryALTinterwordstretchfactor\fontdimen3\font minus
  \fontdimen4\font\relax}
\providecommand{\BIBforeignlanguage}[2]{{%
\expandafter\ifx\csname l@#1\endcsname\relax
\typeout{** WARNING: IEEEtran.bst: No hyphenation pattern has been}%
\typeout{** loaded for the language `#1'. Using the pattern for}%
\typeout{** the default language instead.}%
\else
\language=\csname l@#1\endcsname
\fi
#2}}
\providecommand{\BIBdecl}{\relax}
\BIBdecl

\bibitem{SQUID}
S.~Jacobsson, G.~Durisi, M.~Coldrey, T.~Goldstein, and C.~Studer, ``Quantized
  precoding for massive {MU-MIMO},'' \emph{IEEE Trans. Commun.}, vol.~65,
  no.~11, pp. 4670--4684, Nov. 2017.

\bibitem{WFQ}
A.~Mezghani, R.~Ghiat, and J.~A. Nossek, ``Transmit processing with low
  resolution {D/A}-converters,'' in \emph{Proc. 16th IEEE Int. Conf. Electron.,
  Circuits Syst. (ICECS)}, Dec. 2009, pp. 683--686.

\bibitem{MMSE2}
O.~B. Usman, H.~Jedda, A.~Mezghani, and J.~A. Nossek, ``{MMSE} precoder for
  massive {MIMO} using 1-bit quantization,'' in \emph{Proc. IEEE Int. Conf.
  Acoust., Speech, Signal Process. (ICASSP)}, Mar. 2016, pp. 3381--3385.

\bibitem{reconsider}
O.~De~Candido, H.~Jedda, A.~Mezghani, A.~L. Swindlehurst, and J.~A. Nossek,
  ``Reconsidering linear transmit signal processing in 1-bit quantized
  multi-user {MISO} systems,'' \emph{IEEE Trans. Wireless Commun.}, vol.~18,
  no.~1, pp. 254--267, Jan. 2019.

\bibitem{duality}
K.~U. Mazher, A.~Mezghani, and R.~W. Heath, ``Multi-user downlink beamforming
  using uplink downlink duality with 1-bit converters,'' in \emph{Proc. IEEE
  22nd Int. Workshop on Signal Process. Adv. in Wireless Commun. (SPAWC)},
  2021, pp. 376--380.

\bibitem{MFrate}
Y.~Li, C.~Tao, A.~L.~Swindlehurst, A.~Mezghani, and L.~Liu, ``Downlink
  achievable rate analysis in massive {MIMO} systems with one-bit {DAC}s,''
  \emph{IEEE Commun. Lett.}, vol.~21, no.~7, pp. 1669--1672, Jul. 2017.

\bibitem{ZF}
A.~K. Saxena, I.~Fijalkow, and A.~L. Swindlehurst, ``Analysis of one-bit
  quantized precoding for the multiuser massive {MIMO} downlink,'' \emph{IEEE
  Trans. Signal Process.}, vol.~65, no.~17, pp. 4624--4634, Sept. 2017.

\bibitem{CImodel}
A.~Li, C.~Masouros, F.~Liu, and A.~L. Swindlehurst, ``Massive {MIMO} 1-bit
  {DAC} transmission: A low-complexity symbol scaling approach,'' \emph{IEEE
  Trans. Wireless Commun.}, vol.~17, no.~11, pp. 7559--7575, Nov. 2018.

\bibitem{sep2}
F.~Sohrabi, Y.-F. Liu, and W.~Yu, ``One-bit precoding and constellation range
  design for massive {MIMO} with {QAM} signaling,'' \emph{IEEE J. Sel. Topics
  Signal Process.}, vol.~12, no.~3, pp. 557--570, Jun. 2018.

\bibitem{conference}
Z.~Wu, B.~Jiang, Y.-F. Liu, and Y.-H. Dai, ``A novel negative $\ell_1$ penalty
  approach for multiuser one-bit massive {MIMO} downlink with {PSK}
  signaling,'' in \emph{Proc. IEEE Int. Conf. Acoust., Speech, Signal Process.
  (ICASSP)}, May 2022, pp. 5323--5327.

\bibitem{Bussgang}
J.~Bussgang, ``Crosscorrelation functions of amplitude-distorted {Gaussian}
  signals,'' \emph{RLE Technical Reports}, vol. 216, 1952.

\bibitem{assumption1}
S.~Wagner, R.~Couillet, M.~Debbah, and D.~T.~M. Slock, ``Large system analysis
  of linear precoding in correlated {MISO} broadcast channels under limited
  feedback,'' \emph{IEEE Trans. Inf. Theory}, vol.~58, no.~7, pp. 4509--4537,
  Jul. 2012.

\bibitem{HD}
Y.~M. Lu, ``Householder dice: A matrix-free algorithm for simulating dynamics
  on {Gaussian} and random orthogonal ensembles,'' \emph{IEEE Trans. Inf.
  Theory}, vol.~67, no.~12, pp. 8264--8272, Dec. 2021.

\bibitem{tulino2004random}
A.~Tulino and S.~Verdú, ``Random matrix theory and wireless communications,''
  \emph{Foundations and Trends{\textregistered} in Communications and
  Information Theory}, vol.~1, no.~1, pp. 1--182, 2004.

\bibitem{RZF}
C.~Peel, B.~Hochwald, and A.~Swindlehurst, ``A vector-perturbation technique
  for near-capacity multiantenna multiuser communication-part {I}: {C}hannel
  inversion and regularization,'' \emph{IEEE Trans. Commun.}, vol.~53, no.~1,
  pp. 195--202, Jan. 2005.

\end{thebibliography}

	



\end{document}